\documentclass[a4paper]{article}

\pdfoutput=1

\usepackage[utf8]{inputenc}
\usepackage[T1]{fontenc}
\usepackage{lmodern}
\usepackage{microtype}

\newcommand{\edgeplus}{\ensuremath{\langle + \rangle}}
\newcommand{\edgeminus}{\ensuremath{\langle - \rangle}}

\usepackage[USenglish]{babel}

\usepackage{graphicx} % Required for inserting images
\usepackage{xcolor}
\usepackage{amsmath,amssymb,amsthm}
\usepackage[]{hyperref}
\hypersetup{
    colorlinks,
    linkcolor={red!50!black},
    citecolor={blue!50!black},
    urlcolor={blue!80!black}
}

\definecolor{edgered}{rgb}{0.4, 0.0, 0.0}
\definecolor{edgeblue}{rgb}{0.01, 0.01, 0.4}

\newcommand{\F}{\mathcal{F}}
\newcommand{\C}{\mathcal{C}}
\DeclareMathOperator{\cost}{cost}
\DeclareMathOperator{\weight}{weight}
\DeclareMathOperator{\opt}{opt}

\newtheorem{redrule}{Reduction rule}

\newtheorem{lemma}{Lemma}
\newtheorem{corollary}{Corollary}

\newtheorem{proposition}{Proposition}
\newtheorem{definition}{Definition}
\newtheorem{theorem}{Theorem}

\usepackage{float}
\usepackage[capitalize]{cleveref}
\usepackage{todonotes}
\usepackage{tabularx}
\usepackage[square,numbers]{natbib}
\usepackage{tikz}
\usetikzlibrary{calc}
\usetikzlibrary{positioning,backgrounds,patterns,calc,matrix,shapes,decorations.pathreplacing,decorations.pathmorphing,math}
\tikzset{crossing/.style={cross out, draw=red, minimum size=2*(#1-\pgflinewidth), inner sep=0pt, outer sep=1pt, very thick}, crossing/.default={4pt}}
\usetikzlibrary{arrows.meta,bending,decorations.markings,hobby,patterns,calc}
\newcolumntype{C}[1]{>{\centering\let\newline\\\arraybackslash\hspace{0pt}}m{#1}}

\newcommand{\problem} [1] {\textnormal{\textsc{#1}}}

\newcommand{\cclass} [1] {\textnormal{\textsf{#1}}}

\newcommand{\problembox} [3] {
	\vspace{\dimexpr\parskip+1.3ex}
	\noindent
	\begin{tikzpicture}
		\node[draw=black!40, rounded corners, inner sep=2.3ex] (content) {
			\begin{tabularx} {\dimexpr\columnwidth-5ex} {l X}
				\textbf{Input:} & #2\\
				\textbf{Problem:} & #3\\
			\end{tabularx}
		};

		\node[inner sep=3pt, fill=white, anchor=north west] at ($(content.north west) + (2ex, 1.35ex)$) {\problem{#1}};
	\end{tikzpicture}%
}

\crefname{redrule}{Reduction Rule}{Reduction Rules}

\newcommand{\CCfull}{\problem{Correlation Clustering}}
\newcommand{\CEfull}{\problem{Cluster Editing}}

\newcommand{\CEPVSfull}{\problem{Cluster Editing with Permissive Vertex Splitting}}
\newcommand{\CEPVS}{\CEPVSfull}
\newcommand{\CCVSfull}{\problem{Correlation Clustering with Permissive Vertex Splitting}}
\newcommand{\CCVS}{\CCVSfull}

\newcommand{\MC}{\problem{Multicut}}
\newcommand{\MCVSfull}{\problem{Multicut with Vertex Splitting}}
\newcommand{\MCVS}{\MCVSfull}

\usepackage{authblk}

\title{Correlation Clustering with Vertex Splitting}
\date{}

\author[1]{Matthias Bentert}
\author[2]{Alex Crane}
\author[1]{Pål Grønås Drange}
\author[3]{Felix Reidl}
\author[2]{Blair D.\ Sullivan}

\affil[1]{University of Bergen, Norway. \texttt{Matthias.Bentert@uib.no}, \texttt{Pal.Drange@uib.no}}
\affil[2]{University of Utah, USA. \texttt{alex.crane@utah.edu}, \texttt{sullivan@cs.utah.edu}}
\affil[3]{Birkbeck, University of London, UK. \texttt{f.reidl@bbk.ac.uk}}

\begin{document}

\maketitle

\begin{abstract}
  We explore \problem{Cluster Editing} and its generalization \problem{Correlation Clustering} with a new operation called \emph{permissive vertex splitting} which addresses finding overlapping clusters in the face of uncertain information. We determine that both problems are \cclass{NP}-hard, yet they exhibit significant differences in terms of parameterized complexity and approximability.
  For \problem{Cluster Editing with Permissive Vertex Splitting}, we show a polynomial kernel when parameterized by the solution size and develop a polynomial-time 7-approximation.
  In the case of \problem{Correlation Clustering}, we establish \cclass{para-NP}-hardness when parameterized by the solution size and demonstrate that computing an~$n^{1-\varepsilon}$-approximation is \cclass{NP}-hard for any constant $\varepsilon > 0$.
  Additionally, we extend an established link between \problem{Correlation Clustering} and \problem{Multicut} to the setting with permissive vertex splits.
\end{abstract}

\section{Introduction}
\label{sec:intro}

Discovering clusters, or communities, is a core task in 
understanding the vast amounts of relational data available. 
One limitation of many traditional clustering algorithms is the necessity of 
specifying a desired number of clusters as part of the input. The problem \problem{Cluster
Editing} avoids this by instead aiming to minimize the number of 
edge insertions and removals necessary to transform the input into a \emph{cluster graph}
(a disjoint union of cliques). This problem has been heavily studied by the 
graph-algorithms community and was first proven to be 
fixed-parameter tractable with respect to the number~$k$ of edge modifications by Cai in 1996~\cite{cai1996fixedparametertractability}.
The running time has significantly improved since, with the best known algorithm running in~$O(1.62^k (n+m))$~time~\cite{bocker2012goldenratio}.
The problem also admits a polynomial kernel with~$2k$~vertices~\cite{chen20122kkernel}. 

Formally in \problem{Cluster Editing}, we consider a complete graph where each edge is labeled as positive (which we imagine as colored \emph{\color{edgeblue}blue}) or negative (colored \emph{\color{edgered}red}) and we ask for the minimum number of edges whose color must be changed so that there is a partition of the vertex set where all edges within each part are blue, and all edges between parts are red.
This convention of an edge-labeled complete graph will be useful in our setting and easily maps onto the more common formalism for \problem{Cluster Editing} with an incomplete, uncolored graph as input (imagine the graph edges as blue and its non-edges as red). We also note that other conventions for labelling positive/negative edges exist in the literature, e.g.\ using labels like $\edgeplus$ and $\edgeminus$.

In practice, the positive or negative association between objects is usually computed using a similarity metric which we can think of as an oracle function which, given two objects, computes a score that expresses their (dis)similarity. For large-scale data, the assumption of complete information is then unrealistic for two reasons: First, the quadratic complexity of computing all pairwise associations is prohibitively expensive. Second, the similarity oracle may be
unable to provide a clear answer for certain pairs---suggesting that objects can either be grouped together or kept separate, depending on 
other parts of the data or even external domain context.

Consequently, the case where the input is an incomplete graph with positive and negative labels on the existing edges and no information about pairs not joined by an edge has been considered.
It was introduced by Demaine et al.~\cite{demaine2006correlation} who allowed ``$0$-weight edges'' (zero-edges) in their cluster-editing framework\footnote{In the version discussed by Demaine et al.~\cite{demaine2006correlation}, real weights are assigned to edges, reflecting the certainty level of the oracle in determining the similarity between objects.  We only consider weights in~$\{-1, 0, 1\}$, a common restriction in the literature.}. 
For clarity, we will refer to the problem where zero-edges (non-edges) are allowed as
\problem{Correlation Clustering} and to the problem where the input graph is
\emph{complete}---i.e. every vertex pair is connected either by a blue or a red edge---as~\problem{Cluster Editing}.

The approximability of both \problem{Cluster Editing} and \problem{Correlation Clustering} are well-studied. First considered by Bansal, Blum, and Chawla~\cite{bansal2004correlation}, under the name \emph{correlation clustering}\footnote{There is significant inconsistency in the literature regarding the nomenclature of these problems; as stated, we reserve the name \problem{Correlation Clustering} for the problem where the input is incomplete.}, \problem{Cluster Editing} admits a~$1.437$-approximation~\cite{cao2024understanding} when minimizing the number of disagreements (red edges within and blue edges between clusters). Other variants of \problem{Cluster Editing} which maximize the number of agreements or the correlation (agreements minus disagreements) admit a PTAS (polynomial-time approximation scheme) and a $\Omega(\log n)$-approximation, respectively~\cite{bansal2004correlation, charikar2004maximizing}. In the more general setting of minimizing disagreements for \problem{Correlation Clustering} (i.e., when zero-edges are present but never constitute a disagreement), an $O(\log n)$-approximation is known~\cite{demaine2006correlation}. This result arises from the strong relation between \problem{Correlation Clustering} and \problem{Multicut}\footnote{Given a set of pairs of terminals, $(s_1,t_1), (s_2,t_2), \dots, (s_p, t_p)$, find a set of at most $k$ edges such that after removing these edges, every pair~$(s_i, t_i)$ is disconnected}. The connection was first observed with \problem{Multiway Cut} by 
Bansal, Blum, and Chawla~\cite{bansal2004correlation}, before an approximation-preserving reduction from \problem{Multicut} to \problem{Correlation Clustering} was given independently by both Charikar, Guruswami, and Wirth~\cite{charikar2005clustering} and Demaine et al.~\cite{demaine2006correlation}. The connection to \MC{} also implies that no constant-factor approximation is possible for \problem{Correlation Clustering} unless the Unique Games Conjecture is false~\cite{chawla2006hardness}.

These algorithmic advances provide a positive outlook on applying these clustering variants in practice.
However, the underlying assumption that real-world data segregates into neat, disjoint clusters is often too optimistic as shown by the following domain examples:

\begin{itemize}
\item Document classification: Individual documents often span multiple topics and should therefore belong to multiple topic-clusters;
\item Sentiment analysis: A single piece of text can express very different emotions (e.g.\ sadness mixed with humor);
\item Community detection: Individuals typically
  participate in multiple communities, such as family, professional, and
  hobbyist groups.
\item Language processing: Homonyms like ``bat'' should belong both to an ``animal'' cluster as well as a ``sports-equipment'' cluster.
\end{itemize}

Hence, the emphasis in clustering has recently shifted towards algorithms for 
\emph{overlapping clustering}~\cite{abu-khzam2018cluster,arora2012findingoverlapping,arrighi2023cluster,askeland2022overlappingcommunity,bandyopadhyay2015focsfast,%
baumes2005efficientidentification,bonchi2013overlappingcorrelation,crane2024overlapping,davis2008clearingthefog,du2008overlapping,galbrun2014overlapping,gil-garcia2010dynamichierarchical,goldberg2010findingoverlapping,gregory2007algorithmtofind,wang2011uncoveringoverlapping,wang2010discoveringoverlapping}.
These models move away from the requirement that data must be partitioned into disjoint subsets by considering a variety of definitions for clusters which may intersect. One natural approach is to edit to a more general target graph class (instead of a cluster graph, consider minimizing the number of edge modifications required to achieve some more complex structure that exhibits strong community structure but allows overlap), but it is difficult to define generalizations that align with many applications.

Motivated by this, Abu-Khzam et al.~\cite{abu-khzam2018cluster} proposed an alternative model for
overlapping clustering based on the concept of \emph{splitting} a vertex into two new vertices,  representing an object having two distinct roles within a dataset. 
This approach led to the problem~\problem{Cluster Editing with Vertex
  Splitting}, where edges can be added or deleted, and vertices
can be split. Here, \emph{splitting} a vertex $v$ means replacing it with two copies
$v_1$ and $v_2$ ensuring the union of their (blue) neighbor sets equals the
original vertex's (blue) neighbor set. 
In fact, Abu-Khzam et al.~\cite{abu-khzam2018cluster} propose two different vertex-splitting operations: one (\emph{exclusive} splitting) where $v_1$ and $v_2$ are required to have disjoint (blue) neighborhoods, and another (\emph{inclusive} splitting) where they are allowed to share (blue) neighbors. See~\cref{fig:splits} for an example. Abu-Khzam et al.~\cite{abu-khzam2023cluster} show that \problem{Cluster Editing with Vertex Splitting} is \cclass{NP}-hard and has a $6k$-vertex kernel, where $k$ is the number of edits (edge modifications/vertex splits) allowed. The approximability of this problem remains unknown. 
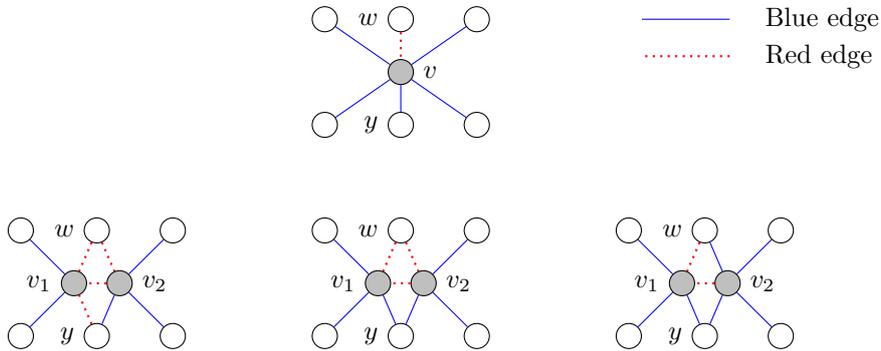
\begin{figure}[t]
    \centering
    \begin{tikzpicture}[yscale=.7]
        \node[circle, draw,label=left:$v_1$,fill=lightgray] at (-4.3,-4) (u1) {};
        \node[circle, draw,label=right:$v_2$,fill=lightgray] at (-3.7,-4) (w1) {} edge[red,dotted,thick](u1);
        \node[circle, draw] at (-5,-3) {} edge[blue](u1);
        \node[circle, draw] at (-5,-5) {} edge[blue](u1);
        \node[circle, draw,label=left:$y$] at (-4,-5) {} edge[blue](w1) edge[red,dotted,thick](u1);
        \node[circle, draw] at (-3,-3) {} edge[blue](w1);
        \node[circle, draw] at (-3,-5) {} edge[blue](w1);
        \node[circle, draw,label=left:$w$] at (-4,-3) {} edge[red,dotted,thick](u1) edge[red,dotted,thick](w1);

        \node[circle, draw,label=right:$v$,fill=lightgray] at (0,0) (v2) {};
        \node[circle, draw] at (-1,1) {} edge[blue](v2);
        \node[circle, draw] at (-1,-1) {} edge[blue](v2);
        \node[circle, draw,label=left:$y$] at (0,-1) {} edge[blue](v2);
        \node[circle, draw] at (1,1) {} edge[blue](v2);
        \node[circle, draw] at (1,-1) {} edge[blue](v2);
        \node[circle, draw,label=left:$w$] at (0,1) {} edge[red,dotted,thick](v2);

        \node[circle, draw,label=left:$v_1$,fill=lightgray] at (-.3,-4) (u2) {};
        \node[circle, draw,label=right:$v_2$,fill=lightgray] at (.3,-4) (w2) {} edge[red,dotted,thick](u2);
        \node[circle, draw] at (-1,-3) {} edge[blue](u2);
        \node[circle, draw] at (-1,-5) {} edge[blue](u2);
        \node[circle, draw,label=left:$y$] at (0,-5) {} edge[blue](w2) edge[blue](u2);
        \node[circle, draw] at (1,-3) {} edge[blue](w2);
        \node[circle, draw] at (1,-5) {} edge[blue](w2);
        \node[circle, draw,label=left:$w$] at (0,-3) {} edge[red,dotted,thick](u2) edge[red,dotted,thick](w2);

        \node[circle, draw,label=left:$v_1$,fill=lightgray] at (3.7,-4) (u3) {};
        \node[circle, draw,label=right:$v_2$,fill=lightgray] at (4.3,-4) (w3) {} edge[red,dotted,thick](u3);
        \node[circle, draw] at (3,-3) {} edge[blue](u3);
        \node[circle, draw] at (3,-5) {} edge[blue](u3);
        \node[circle, draw,label=left:$y$] at (4,-5) {} edge[blue](w3) edge[blue](u3);
        \node[circle, draw] at (5,-3) {} edge[blue](w3);
        \node[circle, draw] at (5,-5) {} edge[blue](w3);
        \node[circle, draw,label=left:$w$] at (4,-3) {} edge[red,dotted,thick](u3) edge[blue](w3);

        \node[circle, draw,color=white,opacity=0] at (3, 1) (legendv1) {};
        \node[circle, draw,color=white,opacity=0] at (4.5, 1) (legendv2) {} edge[blue,label=right:Blue edge](legendv1);
        \node[circle, draw,color=white,opacity=0] at (3, .3) (legendv3) {};
        \node[circle, draw,color=white,opacity=0] at (4.5, .3) (legendv4) {} edge[red,dotted,thick,label=right:Red edge](legendv3);
    \end{tikzpicture}
\caption{A vertex $v$ in an (incomplete) correlation graph (top). The bottom row gives toy examples of
exclusive (left), inclusive (center), and permissive (right) vertex splits of~$v$ into~$v_1$ and~$v_2$. For clarity,
some red edges incident to~$v_1$ and~$v_2$ are omitted from each figure on the bottom row.}
\label{fig:splits}
\end{figure}%

A significant limitation of both existing notions of vertex splitting is that they require red edges to be preserved by both copies of a split vertex.  For example, consider a red edge~$uv$ in data arising from word classification, where~$u$ and~$v$ correspond to ``bat'' and~``cat'', respectively. It could be that the edge was produced by our oracle as a result of ``bat'' being
interpreted as a piece of sports equipment, not an animal.  However, when ``bat'' is split so that each meaning has its own vertex, we wish to retain the red edge only on one of the copies of $v$ (the one \emph{not} corresponding to the small flying mammal, as this does have similarities with a cat). Motivated by this, we introduce a new operation called \emph{permissive
  vertex splitting} which allows replacing a vertex $v$ with two copies~$v_1$ and~$v_2$ with the restriction that if~$uv$
is a blue edge (or red edge, respectively), then at least one of~$uv_1$ and~$uv_2$ is a blue edge
(red edge, respectively). Beyond that, we are free to choose what to do with the newly-created neighborhoods.
We call the new problem variant, where edges can be added or deleted and vertices
can be permissively split, \CCVSfull. We show that sequences of permissive vertex splits solving this problem correspond directly to a natural notion of overlapping clustering (see~\cref{def:overlapping-clustering}), adding to the motivation for this definition of splitting. 

Extending the prior work relating \problem{Correlation Clustering} to \MC{}, we show that \CCVS{} can be reduced to the new problem \MCVSfull{} and vice versa, meaning that the computational complexities of these problems are
essentially the same.  We then show that \MCVS{}, and hence also \CCVS{}, are
\cclass{para-NP}-hard (with respect to solution size), and \cclass{NP}-hard to approximate within an $n^{1 - \epsilon}$
factor for any~$\epsilon > 0$.
Because of the inherent hardness of \CCVS{}, we then turn our attention to the
setting where there are no zero-edges, i.e., to \CEPVSfull{}.  
We show that this problem remains \cclass{NP}-hard, but on the positive side admits a polynomial
kernel (and thus is fixed-parameter tractable).  Finally, we give a polynomial-time algorithm
which provides a 7-approximation for \CEPVS{}.

\section{Preliminaries}\label{sec:prelims}

We refer the reader to the textbook by Diestel~\cite{diestel2012graph} for standard graph-theoretic definitions and notation.
A \emph{star} is a tree with exactly one internal vertex. In particular, a star has at least two leaves.
A \emph{red clique} is a clique in which all edges are red. A \emph{blue clique} is defined similarly. 
For a positive integer $n$, we denote by~$[n] = \{1, 2, \ldots n\}$ the set of all positive integers up to~$n$.
An \emph{incomplete correlation graph} is a simple, unweighted, and undirected graph $G = (V, B, R)$ with two disjoint edge relations $B$ (blue) and $R$ (red). If such a graph is complete, i.e., $B \cup R = \binom{V}{2}$, then we call it a \emph{correlation graph}. For a vertex $v \in V$ we write~$N^R(v)$ to denote the set of neighbors adjacent to~$v$ via red edges (\emph{red neighbors}) and~$N^B(v)$ for those adjacent via blue edges (\emph{blue neighbors}). A \emph{cluster graph} is a correlation graph in which the blue edges form vertex-disjoint cliques (and thus all edges between the cliques are red).
We can now formally define our vertex-splitting operation.

\begin{definition}
    \label{def:permissive-vertex-split}
    A \emph{permissive vertex split} of a vertex $v$ in an (incomplete) correlation graph~$G$ is the replacement of $v$ in $G$ with
    two new vertices~$v_1$ and~$v_2$ such that
    \begin{itemize}
        \item $N^R(v) \subseteq N^R(v_1) \cup N^R(v_2)$, and
        \item $N^B(v) \subseteq N^B(v_1) \cup N^B(v_2)$.
    \end{itemize}
\end{definition}
In other words, we create a new graph where every red (blue) neighbor of~$v$ is a red (blue)
neighbor of at least one of $v_1$ or $v_2$.
All other edges incident to $v_1$ and $v_2$ can be chosen arbitrarily.
In particular, in incomplete correlation graphs, we can assume that all these other edges are neutral (i.e., the ``edges'' do not exist), while in correlation graphs, it is usually simpler to make these edges either red or blue to keep the graph complete.
Notably, the edge~$v_1v_2$ can always be assumed to be a red edge as splitting a vertex into two vertices that end up in the same (blue) connected component is never advantageous.
For the remainder of this text, unless otherwise specified all vertex splits are permissive. Given
a sequence~$\sigma = (\sigma_1, \sigma_2, \ldots, \sigma_k)$ of $k$ vertex splits performed on an (incomplete) correlation graph,
we denote the resulting (incomplete) correlation graph by $G_{|\sigma}$. Each vertex $u$ in $G_{|\sigma}$ corresponds to exactly
one vertex $v$ in $G$. We say that $v$ is $u$'s \emph{ancestor}, and that $u$ is a \emph{descendant} of $v$. If $u = v$, then~$v$ is an \emph{unsplit} vertex. Otherwise we say that $v$ is a \emph{split} vertex and that $u$ is the descendant of a split vertex.

\begin{definition}
    An \emph{erroneous cycle} is a simple cycle that contains exactly one red edge. An (incomplete) correlation graph~$G$ \emph{contains} an erroneous cycle if it contains a subgraph that is an erroneous cycle. A \emph{bad triangle} is an erroneous cycle of length 3.
\end{definition}

Erroneous cycles are the canonical obstruction in \CCfull{}~\cite{charikar2005clustering,demaine2006correlation}, and bad triangles are the canonical obstruction in \CEfull{}. Usually, these problems are formulated as edge editing problems, i.e., delete a minimum number of edges (blue or red) such that the resulting graph has no erroneous cycles/bad triangles.
Previous work on \CEfull{} with (inclusive or exclusive) vertex splitting has allowed both edge edits and vertex splits as editing operations~\cite{abu-khzam2018cluster,abu-khzam2021greedy,abu-khzam2023cluster,arrighi2023cluster}.
However, we note that permissive vertex splitting is flexible enough to capture all editing operations.
First, note that in the setting with blue and red edges, each edge-editing operation can be seen as changing the color of an edge.
Now, consider any solution~$\sigma$ in which the color of an edge~$uv$ is changed.
Then, we construct a new sequence of the same length where this edge edit is replaced by a vertex split.
We choose one endpoint (without loss of generality $v$) and split it into~$v_1$ and~$v_2$.
The neighborhood of~$v_1$ is exactly the neighborhood of the initial vertex~$v$ except that the edge towards~$u$ has the other color.
If the edge~$uv$ was initially red, then the vertex~$v_2$ has all vertices in the graph as red neighbors.
If the edge~$uv$ was blue, then we add blue edges between~$v_2$ and all vertices that end up in the same (blue) connected component as (one descendant of)~$u$ in~$G_{\mid \sigma}$.
The result of the edge edit is now modeled exactly by~$v_1$ and the operation is safe because~$v_2$ cannot participate in any erroneous cycle as it is a twin of (one descendant of)~$u$.
Moving forward, we assume that all editing operations are vertex splits, and we say that a sequence~$\sigma$ of vertex splits \emph{clusters} an (incomplete) correlation graph~$G$ if $G_{|\sigma}$ has no erroneous cycles.
Formally, we study the following problem.

\problembox{\CCVSfull{}}
{An (incomplete) correlation graph $G$ and a non-negative integer $k$.}
{Does there exist a sequence $\sigma$ of at most $k$ vertex splits which clusters $G$?}

\CEPVSfull{} is the same problem restricted to correlation graphs.
We conclude this section with our main structural insight, stating that clustering
an (incomplete) correlation graph $G$ via a sequence of (permissive) vertex splits is equivalent to performing a very natural notion of overlapping clustering on the vertices of $G$.

\begin{definition}
    A \emph{covering} of an (incomplete) correlation graph $G=(V,E)$ is a set family~$\F \subseteq 2^{V}$
    such that~$\bigcup \F = V$. The \emph{cost} of the covering $\F$ is
    \[
        \cost_G(\F) = \sum_{v \in V} (\#\F(v) - 1),
    \]
    where $\#\F(v) := |\{ X \mid v \in X \in \F\}|$ counts the number of sets in $\F$ which contain $v$.
\end{definition}

\begin{definition}
\label{def:overlapping-clustering}
    An \emph{overlapping clustering} of an (incomplete) correlation graph~$G$ is a covering~$\F$ with the following two properties:
    \begin{itemize}
        \item for every blue edge $uv \in B$, there exists at least one cluster~$X \in \F$ with~$\{u,v\} \subseteq X$, and
        \item for every red edge $uv \in R$, there exists two distinct clusters $X,Y \in \F$ with $u \in X$ and~$v \in Y$.
    \end{itemize}
    For a specific edge $uv$, we say that a clustering \emph{covers} the edge if it is blue and the first condition holds and we say that it \emph{resolves} the edge if it is red and the second condition holds.
\end{definition}

\begin{lemma}
	\label{lemma:vs-equivalent-to-oc}
    An (incomplete) correlation graph $G$ can be clustered with $k$ vertex splits
    if and only if $G$ has an overlapping clustering of cost $k$.
\end{lemma}
\begin{proof}
    For the first direction, let $\sigma$ be a sequence of $k$ vertex splits
    clustering~${G = (V, B, R)}$, i.e.,~$\sigma$ produces a graph $G_{|\sigma} = (V_{|\sigma}, E_{|\sigma})$ with no erroneous cycles.
    We will
    construct an overlapping clustering $\F$ of cost at most $k$. We note that it is easy to
    extend any such overlapping clustering to one of cost exactly $k$.
    We begin by choosing an
    arbitrary vertex~$v \in V_{|\sigma}$. We denote by $v^*$ the ancestor of $v$ in $V$.
    Let $C_{v} \subseteq V_{|\sigma}$ be the vertices of the
    connected component of $v$ in the subgraph of $G_{|\sigma}$ induced by all blue edges,
    and $C_{v^*}$ be the set of corresponding ancestor vertices in $V$.
    We add $C_{v^*}$ to $\F$ and remove $C_v$ from~$G_{|\sigma}$. We repeat
    this process exhaustively. The resulting $\F$ is a covering of $G$, as each vertex in
    $V$ has at least one descendant in $V_{|\sigma}$. Moreover, our construction guarantees
    that each vertex in~$V_{|\sigma}$ is considered exactly once. Consequently, for each vertex~$v \in V$ we
    have that $\#\F(v)$ is no greater than the number of descendants of $v$ in $V_{|\sigma}$.
    Thus, $\F$ has cost at most~$k$. Each blue edge is covered by construction.

    For the final step, we show
    how to augment~$\F$ such that all red edges are resolved while maintaining that $\cost_G(\F) \leq k$. We begin by
    identifying some red edge $uv$ which is not resolved by $\F$. This implies that each of $u$ and $v$
    are contained in exactly one cluster $X \in \F$. The red edge $uv$ implies that there is some red edge $u_1v_1$
    in $G_{|\sigma}$, where $u_1$ is a descendant of~$u$ and $v_1$ is a descendant of $v$. Moreover, the
    construction of $\F$ guarantees that there is some blue path between $v_1$ and a descendant of $u$, but
    this latter descendant cannot be~$u_1$ or else we have identified an erroneous cycle in $G_{|\sigma}$. Thus,~$u$
    has multiple descendants in~$G_{|\sigma}$ and is therefore a split vertex. Since $u$ is a split vertex but is only contained in one cluster~$X$ in $\F$, we can
    add the cluster~$\{u\}$ to $\F$, thereby resolving $uv$, while maintaining that $\cost_G(\F) \leq k$.
    We repeat this process until all red edges are resolved.

    For the other direction, let $\F$ be an overlapping clustering of $G$ with cost~$k$.
	For each vertex~$v$ that is contained in more than one cluster set in~$\F$, we split~$v$ a total of~${\#\F(v)-1}$~times. We assign each descendant to one set~$X \in \F$ with~$v \in X$ and we create blue edges towards all other vertices that are contained in~$X$ (or to the specific descendant of a vertex in~$X$ that was also assigned to~$X$).
	All other edges incident to the descendant of~$v$ are red.
	We first show that this construction indeed corresponds to a series of vertex splits.
	For each blue edge~$uv$, we have that there is some cluster set~$X \in \F$ with~$u,v \in X$.
	Hence, if~$u$ and/or~$v$ are split, then the blue edge~$uv$ corresponds to the blue edge between the two copies of~$u$ and~$v$ that are assigned to~$X$.
	For each red edge~$uv$, we have that there are some cluster sets~$X \neq Y \in \F$ with~$u \in X$ and~$v \in Y$.
	Hence, if~$u$ and/or~$v$ are split, then the red edge~$uv$ corresponds to the red edge between (the descendant of)~$u$ assigned to~$X$ and (the descendant of)~$v$ that is assigned to~$Y$.
	Moreover, we did exactly~$\cost(\F)$ splits.

	It remains to show that the sequence of splits results in a graph that does not contain any erroneous cycle.
	Assume towards a contradiction that an erroneous cycle~${(u=v_0,v_1,\ldots,v_p=w,u)}$ with red edge~$uw$ remains.
	Note that each vertex is assigned to exactly one cluster set in~$\F$ as each unsplit vertex is contained in exactly one set in~$\F$ and each descendant of a split vertex is assigned to a cluster set by construction.
	We will show that there is no blue edge between vertices that are assigned to different clusters and no red edge between vertices that are assigned to the same cluster set.
	This finishes the proof as~$u$ and~$w$ are then assigned to different cluster sets as they share a red edge, but~$w_i$ and~$w_{i-1}$ are assigned the same cluster set for each~$i \in [p]$, a contradiction.
	First, assume that there is a blue edge~$xy$ where~$x$ and~$y$ are assigned to different cluster sets.
	If~$x$ and~$y$ are both unsplit vertices, then the blue edge between them is not covered by~$\F$, a contradiction.
	Hence, at least one of the two vertices is the descendant of a split vertex and by construction, all edges to vertices that are assigned to different cluster sets are red.
	Now assume that there is a red edge~$xy$ where~$x$ and~$y$ are assigned to the same cluster set~$X \in \F$.
	Again, if~$x$ and~$y$ are both unsplit vertices, then they are only contained in~$X$ in~$\F$ and hence the red edge between them is not resolved by~$\F$, a contradiction.
	So at least one of the two vertices is the descendant of a split vertex and, by construction, all edges to vertices that are assigned to~$X$ are blue, a final contradiction.
	This concludes the proof.
\end{proof}

\section{Incomplete Information}\label{sec:incomplete-information}

We first consider the more general problem \CCVSfull{} which allows for incomplete information.
Without vertex splits, it has long been known that \CCfull{} is in fact
equivalent to \MC{}~\cite{demaine2006correlation}, which is the problem of deleting a minimum number of
edges from a graph~$G = (V, E)$
such that every \emph{terminal pair} of distinct vertices in a set~$S \subseteq \binom{V}{2}$ is separated in the resulting graph. We define \MCVSfull{} and show that it is equivalent to \CCVS{}. We believe that this result is of independent interest, but it will also
prove immediately useful as it facilitates the main results of this section. Specifically, \CCVS{} and
\MCVS{} are both \cclass{para-NP}-hard when parameterized by the number of vertex splits, and for any~$\varepsilon > 0$ it is
\cclass{NP}-hard to approximate either problem within a factor of~$n^{1 - \varepsilon}$.

First, we define our new \MC{} variant. In this context, we use standard graph terminology,
i.e., we discuss simple, unweighted, and undirected graphs with a single edge relation~$E$. Note that this is equivalent to
a correlation graph where edges in $E$ are blue and all other vertex pairs are red, so permissive vertex splits
are still well-defined. However in the \MC{} context, we can safely assume that all vertex splits are
\emph{exclusive}, i.e., whenever splitting a vertex~$v$ into descendants $v_1$ and $v_2$ we have that~${N(v_1) \cup N(v_2) = N(v)}$ and~${N(v_1) \cap N(v_2) = \emptyset}$. The reason is that in \MC{} it is never advantageous to
assign more edges than required.
Note that in the classic version of \MC, it does not make sense to have an edge between two vertices of a terminal pair.
We decided to keep this restriction as it streamlines some of the following arguments.
A related technical detail to discuss is what happens to a terminal pair when one of its two vertices is split.
We work with the variant where the terminal pair is simply removed in this case.
Note that this is equivalent to the variant where we can choose either of the descendants to replace the original vertex in the terminal pair,
since, as previously mentioned, we may always assume that any two descendants of the same vertex end up in different connected components.
The formal definition of \MCVSfull{} is hence as follows.

\problembox{\MCVSfull{}}
{A graph $G = (V, E)$, an integer $k$, and a set $S \subseteq \binom{V}{2}$ of \emph{terminal pairs} with~$S \cap E = \emptyset$.}
{Does there exist a sequence $\sigma$ of at most $k$ (exclusive) vertex splits such that each pair in $S$ is separated in $G_{|\sigma}$?}

We now show that \CCVS{} and \MCVS{} are equivalent problems.
Let $(G = (V, B, R), k)$ be an instance of \CCVS{}. We construct an equivalent instance~${(H = (V', E'), S, k)}$ of \MCVS{}
as follows.
For each vertex $v \in V$ we create a vertex $v'$ in $V'$.
Additionally, for each blue edge $uw \in B$ we add the edge $u'w'$ to $E'$.
Finally, for each red edge $uw \in R$ we add the terminal pair $\{u', w'\}$ to $S$. This completes the construction of~$H$.

\begin{theorem}\label{thm:CCVS-to-MCVS}
    For any integer $k \geq 0$, $(G,k)$ is a yes-instance of \CCVSfull{} if and only if~$(H,S,k)$ is a yes-instance of \MCVSfull.
\end{theorem}
\begin{proof}
    For the first direction, let $\sigma = (\sigma_1, \sigma_2, \ldots)$ be a sequence of vertex splits clustering~$G$. We will construct a sequence $\sigma'$ of the same length which separates each pair in~$S$ by considering each~$\sigma_i$ in order. If~$\sigma_i$ splits vertex~${v \in V}$ into~$v_1$ and~$v_2$
    then $\sigma'_i$ splits~$v'$ into~$v'_1$ and~$v'_2$. By construction, each neighbor~$u'$ of~$v'$ corresponds to a blue
    neighbor $u$ of $v$. If $u$ is a blue neighbor of $v_1$, then we create the edge $v'_1u'$. Otherwise, we create the
    edge $v'_2u'$. This completes the construction of $\sigma'$.
    Now, we assume toward a contradiction that some terminal pair~$\{v', u'\}$ is connected in $H_{|\sigma'}$. Then, there is some path~${(v' = w'_0, w_1, \ldots, w'_p = u')}$ in~$H_{|\sigma'}$. Note that our construction ensures that this path contains at least
    two edges, and that there is a corresponding blue path $(v = w_0, w_1, \ldots, w_p = u)$ in~$G_{|\sigma}$. Moreover, because
    $\{v', u'\}$ is a terminal pair in $H_{|\sigma'}$, $vu$ is a red edge in~$G_{|\sigma}$. Thus, we have
    identified an erroneous cycle~$(v = w_0, w_1, \ldots, w_p = u, v)$ in~$G_{|\sigma}$, contradicting that $\sigma$
    clusters~$G$.

    For the other direction, let $\sigma' = (\sigma'_1, \sigma'_2, \ldots)$ be a sequence of vertex splits such
    that no terminal pair is connected in $H_{|\sigma'}$. As before, we will construct a solution $\sigma$ of the same length
    by considering each $\sigma'_i$ in order. If~$\sigma'_i$ splits~$v'$ into~$v'_1$ and $v'_2$, then we split
    the corresponding vertex~$v$ into $v_1$ and $v_2$ as follows. If $v'$ is a terminal with partner $u'$, then
    our construction guarantees that~$vu$ is a red edge in $G$.
    We create the red edge~$v_1u$ if $v'_1$ (or one of its descendants) is in a different component from~$u$ (or one of its descendants) in~$H_{|\sigma'}$.
    Otherwise, we create the red edge~$v_2u$. We mark the relevant pair of descendants so that, when performing
    subsequent splits, the red edge is always assigned such that its endpoints in $G_{\sigma}$ correspond to vertices in different connected components of $H_{|\sigma'}$. 
    Next, for each $u' \in N(v')$, we create the blue edge~$v_1u$ if $\sigma'_i$ assigns $u'$ to $N(v'_1)$. 
    Otherwise, we create the blue edge~$v_2u$.
    We now assume toward a contradiction that there is an erroneous cycle $(v = w_0, w_1, \ldots, w_p = u, v)$ in
    $G_{|\sigma}$, with $vu$ being the red edge. The blue path $(v = w_0, w_1, \ldots, w_p = u)$ guarantees that there is a
    path~$(v' = w_0', w'_1,\ldots, w_p' = u')$ from $v'$ to $u'$ in~$H_{|\sigma'}$.
    Note that the red edge~$vu$ implies
    that~$\{v', u'\}$is a terminal pair. This contradicts that no terminal pair is connected in~$H_{|\sigma'}.$
\end{proof}

To reduce \MCVS{} to \CCVS{},
we simply reverse the previous reduction of \CCVS{} to~\MCVS{}. Formally,
let~${(G = (V, E), S, k)}$ be an instance of \MCVS. We create an instance~${(H = (V', B, R), k)}$ of
\CCVS{} as follows. For each vertex $v \in V$, we add vertex $v'$ to $V'$, for each edge~$uw \in E$, we add the blue edge $u'w'$ to $B$, and for each terminal pair~$\{u, v\} \in S$, we add
the red edge $u'v'$ to $R$. Note that $B$ and $R$ are disjoint, as by definition no terminal pair in $S$ is also an edge in $E$.

\begin{theorem}\label{thm:MCVS-to-CCVS}
    For any integer $k \geq 0$, $(G, S, k)$ is a yes-instance of~\MCVSfull{} if and only if $(H, k)$ is a yes-instance of \CCVSfull{}.
\end{theorem}
\begin{proof}
    Note that applying the reduction behind \cref{thm:CCVS-to-MCVS} to~$H$ results in the instance~$(G, S, k)$.
    Thus, \cref{thm:CCVS-to-MCVS} already shows that the two instances are equivalent.
\end{proof}

Theorems \ref{thm:CCVS-to-MCVS} and \ref{thm:MCVS-to-CCVS}, together with the observation that both
reductions exactly preserve the number of vertices, allow us to state the following strong notion of equivalence
between \CCVS{} and \MCVS{}.

\begin{corollary}\label{cor:CCVS-MCVS-equivalent}
    For any function $f$, \CCVS{} admits a kernel of size $f(k)$ if and only if~\MCVS{} does. Furthermore,
    the minimization variant of \CCVS{} admits a polynomial-time $f(n)$-approximation algorithm
    if and only if the minimization variant of \MCVS{} does.
\end{corollary}

Now that we have established the equivalence of \MCVS{} and \CCVS{}, we are ready to
show the hardness of both problems.

\begin{theorem}\label{thm:MCVS-para-NP-hard}
    \MCVS{} is \cclass{NP}-hard even if~$k=2$.
    Additionally, for any $\varepsilon > 0$ it is
    \cclass{NP}-hard to approximate \MCVS{} to within a factor of~$n^{1-\varepsilon}$.
\end{theorem}
\begin{proof}
    Let~$G = (V, E)$ be the input graph for $k$-\textsc{Colorability} with $k \geq3$. We will construct an equivalent input instance $(H,S,k-1)$
	for \MCVS. We construct the graph $H$ and terminal set $S$ from $G$ as follows. We add $V$ to $H$ and for each edge $uv \in E$, we add the pair~$\{u, v\}$ to $S$. We then add a new vertex $a$ to $H$, and create edges from $a$ to all other vertices. This completes the construction.

    We first argue that we may assume that any solution of $(H, S, k - 1)$ \emph{only} splits~$a$. To see this, let $\sigma$ be a sequence of vertex splits of length at most~$k-1$ such that all terminal pairs are disconnected in $H_{|\sigma}$. Suppose that some vertex~$v \neq a$ is split. Before this split, $v$ only has one neighbor~$a^*$, which is either equal to $a$ or a descendant of $a$. We simply replace the split of $v$ with a split of $a^*$ into $a^*_1$ and $a^*_2$ such that $N(a^*_1) = \{v\}$ and $N(a^*_2) = N(a^*)\setminus \{v\}$. In the resulting graph, $v$ is disconnected from all other vertices in $V$, and so it is disconnected from all of its terminal partners. We proceed with the assumption that in any solution $a$ is the only split vertex.

    Assume that~$(H,S,k-1)$ is a yes-instance. We show that then~$G$ is $k$-colorable. By the above, $H_{|\sigma}$ contains~$k$ descendants~$a_1,\ldots,a_k$ of~$a$ and these vertices naturally partition the set~$V$ into~$k$ sets~$C_i = N(a_i)$ for~$i \in [k]$. No terminal pair can appear with both endpoints in one of these sets so the same holds for~$E = S$. Hence, $C_1,\ldots,C_k$ is a valid $k$-coloring of~$G$.

    In the other direction, assume that~$G$ has a $k$-coloring with the color partition $C_1,\ldots,C_k$. Then we can split~$a \in H$ a total of~$k-1$ times into descendants~$a_1,\ldots,a_k$ such that~$N(a_i) = C_i$. Since~$\{a_1,\ldots,a_k\}$ is independent it is easy to verify that these~$k-1$ splits separate every terminal pair in~$S$.

    We conclude that \MCVS{} is already \cclass{NP}-hard with parameter~$k = 2$
	as the above provides a reduction from $3$-\textsc{Colorability}. The approximation hardness follows directly from the facts that, given any constant~$\varepsilon>0$, computing an~$n^{1-\varepsilon}$-approximation for \problem{Chromatic Number} is \cclass{NP}-hard~\cite{zuckerman2007linear}, and that our constructed instance of \MCVS{} has only $n+1$ vertices.
\end{proof}
Taken together with~\cref{cor:CCVS-MCVS-equivalent},~\cref{thm:MCVS-para-NP-hard} gives us the same result
for \CCVS{}.
\begin{corollary}\label{cor:CCVS-para-NP-hard}
    \CCVS{} is \cclass{NP}-hard even if~$k=2$.
    Additionally, for any $\varepsilon > 0$ it is
    \cclass{NP}-hard to approximate \CCVS{} to within a factor of~$n^{1-\varepsilon}$.
\end{corollary}

\section{Complete Information}

We now focus our attention on correlation graphs, i.e., we study \CEPVSfull{}. Our main results are
\cclass{NP}-hardness (\cref{sec:np-hardness}), a polynomial kernel (\cref{sec:poly-kernel}),
and a polynomial-time 7-approximation (\cref{sec:approximation}). We begin by introducing a new structure
and subsequent lemmas which will be helpful in attaining the latter two results.

\begin{definition}
    A \emph{bad star}~$S$ in a correlation graph $G$ is a set~$\{v_0,v_1,\ldots,v_{|S|-1}\}$ of vertices where all edges in~$\{\{v_0,v_i\}\mid i \in [|S|-1]\}$ are blue and all edges in~$\{\{v_i,v_j\}\mid i \neq j \in [|S|-1]\}$ are red. The vertex~$v_0$ is called the center and all other vertices are called leaves. The \emph{weight} of a bad star $\weight(S)$ is the number of leaves in the star minus one.
    A \emph{bad star forest} is a collection~$T$ of vertex-disjoint bad stars. We write $\weight(T) := \sum_{S \in T} \weight(S)$ to denote the sum of weights of its members.
    A correlation graph $G$ \emph{contains} a bad star forest if it contains a subgraph which is a bad star forest.
\end{definition}

The first lemma states a useful lower bound in terms of bad stars.

\begin{lemma}\label{lemma:bad-star-lower}
    If $G$ contains a bad star forest of weight~$k$ then we need at least~$k$ vertex splits to cluster~$G$.
\end{lemma}
\begin{proof}
    We begin by showing that if $G = (V, B, R)$ contains a (not necessarily induced) subgraph $H = (V_H, B_H, R_H)$
    and at least $k$ vertex splits are needed to cluster $H$, then at least $k$
    vertex splits are needed to cluster~$G$. Suppose otherwise. Then, using~\cref{lemma:vs-equivalent-to-oc}, there is some overlapping clustering $\F$ of~$G$
    with cost less than $k$.
    We will construct an overlapping clustering $\F_H$ of $H$ with cost less than~$k$.
    We begin by setting~$\F_H = \{X \cap V_H \ | \ X \in \F\}$. It is clear that this is a covering of
    $H$, that every blue edge is covered, and that~$\#\F_H(v) \leq \#\F(v)$ for every vertex $v \in V_H$.
    We now ensure that each red edge is resolved. Let~$uv$ be a red edge
    which is not resolved, so each of $u$ and $v$ belong to only a single cluster~$X \in \F_H$. In this case,
    we claim that~$\F$ contains two distinct clusters~$Y \neq Z$ such that~$Y \cap V_H = Z \cap V_H = X$. Otherwise,
    either $\F$ does not resolve~$uv$ or one of~$u$ or~$v$ is contained in multiple clusters of
    $\F_H$, both contradictions. Thus, we can safely add the cluster $\{u\}$ (chosen without loss of generality) to $\F_H$,
    thereby resolving $uv$ while maintaining that~$\#\F_H(u) \leq \#\F(u)$.
    We repeat this process iteratively until all red edges are resolved. In doing so, we produce an
    overlapping clustering $\F_H$ of $H$ with~cost$_H(\F_H) \leq$ cost$_G(\F) < k$, a contradiction.

    It remains to show that a bad star forest
    of weight $k$ requires at least $k$ vertex splits to cluster. We begin by showing that
    a bad star $S$ of weight $k$ requires at least $k$ vertex splits.
    Let~$x$ be the center vertex of~$S$ and let~$\F$ be an overlapping clustering of $S$.
    Let the clusters in~$\F$ which contain~$x$ be~${C_1, C_2, \ldots, C_p}$.
    Moreover for each~$1 \leq i \leq p$,~let $\hat{C}_i = C_i \setminus \{x\}$.
    Note that we may assume each $\hat{C}_i$ is nonempty, as the cluster $\{x\}$ covers no blue edges and resolves
    no red edges in~$S$, and can therefore be safely removed from $\F$.
    Observe also that each leaf~$v$ of~$S$ must be contained in some set~$\hat{C}_i$ since the edge~$xv$ is blue.
    Now suppose that some leaf~$v$ is contained in two sets~$\hat{C}_i \neq \hat{C}_j$.
    We remove $v$ from the cluster~$C_j$ (chosen arbitrarily) and add the cluster $\{v\}$ to $\F$. The blue
    edge~$xv$ is still covered by~$C_i$, the cluster~$\{v\}$ ensures that all red edges incident to~$v$
    are still resolved, and we have not increased the cost of the clustering. Thus, we may safely
    assume that the sets $\hat{C}_1, \hat{C}_2, \ldots \hat{C}_p$ are a partition of the leaves of $S$. Consider
    one such set $\hat{C}_i$. These leaves induce a red clique and none of these red edges is resolved by~$C_i$, so we have that at least $|\hat{C_i}| - 1$ of these leaves are contained in multiple clusters in~$\F$.
    Since we also know that~$\#\F(x) = p$, we conclude
    \begin{align*}
    	\cost_S(\F) &\geq (|\hat{C}_1|-1) + (|\hat{C}_2|-1) + \ldots + (|\hat{C_p}|-1) + \#\F(x) - 1 \\
    	&= |\hat{C}_1| + |\hat{C}_2| + \cdots + |\hat{C}_p| - p + p - 1 = |S \setminus \{x\}| - 1 = \weight(S) = k
    \end{align*}
    Finally, let
    $T$ be a bad star forest made up of $t$ bad stars $S_1, S_2, \ldots, S_t$. Let~$k$ be the weight of
    $T$ and suppose toward a contradiction that $T$ admits an overlapping clustering~$\F$ of cost less than $k$. Then,
    we repeat the technique from earlier in this proof to construct overlapping clusterings $\F_{S_1}, \F_{S_2}, \ldots \F_{S_t}$ of the bad stars. Because
    the bad stars are vertex-disjoint, we have that ${\text{cost}_{S_1}(\F_{S_1}) + \text{cost}_{S_2}(\F_{S_2}) + \ldots + \text{cost}_{S_t}(\F_{S_t}) \leq \text{cost}_{T}(\F) < k}$.
    This implies that there is some~$S_i$ such that $\text{cost}_{S_i}(\F_{S_i})$ is less than the weight of $S_i$,
    but we have already proven that this is impossible.
\end{proof}

The second lemma states that every optimal solution contains a cluster that contains all vertices of a sufficiently large blue clique.

\begin{lemma}\label{lemma:big-cluster2}
    If a correlation graph~$G$ contains a blue clique~$C$ of size at least~$k+1$, then any overlapping clustering $\F$ of cost at most $k$ contains a set~$X \in \F$ with~$C \subseteq X$.
\end{lemma}
\begin{proof}
    Let~$C$ be a blue clique in~$G$ of size at least~$k+1$ and let $\F$ be any overlapping clustering of cost at most~$k$.
    Assume towards a contradiction that~$\F$ does not contain a set~$X$ with~$C \subseteq X$.
    Since~$\F$ has cost at most~$k$, there exists a vertex~$v \in C$ that is contained in exactly one set~$Y \in \F$.
    Moreover, since $Y$ does not contain all vertices of~$C$ by assumption, there exists a vertex~$u \in C \setminus Y$.
    Since~$C$ is a blue clique, the edge~$uv$ is blue and is covered by some set~$Z \in \F$.
    Observe that~$Y \neq Z$ as~$u \in Z$ and~$u \notin Y$.
    Since~$Z$ covers the edge~$uv$, it holds that~$v \in Z$, a contradiction to the assumption that~$v$ is only contained in~$Y$.
\end{proof}

\subsection{NP-hardness}\label{sec:np-hardness}

We now show that \CEPVSfull{} is \cclass{NP}-hard.

\begin{proposition}
	Deciding whether a given correlation graph admits an overlapping clustering of cost at most~$k$ is \cclass{NP}-hard.
\end{proposition}

\begin{proof}
	We reduce from \textsc{Vertex Cover}.
	Let~$(G=(V,E),k')$ be an instance of \textsc{Vertex Cover}.
	We construct a correlation graph~$H$ as follows.
	Let~$U$ be a set of~$k'+1$ vertices (not contained in~$V$).
	The vertex set of~$H$ is~$U \cup V$.
	For each edge~$e=uv \in E$, we add a red edge~$uv$ to~$H$.
	All other edges (including all edges incident to a vertex in~$U$) are blue.
	Finally, we set~$k = k'$.

	We next show that the reduction is correct.
	First assume that~$G$ contains a vertex cover~$S$ of size at most~$k$.
	We construct an overlapping clustering~$\F$ of~$H$ of cost at most~$k$ as follows.
	The family~$\F$ contains one set~$X = U \cup V$ and for each vertex~$v \in S$, it contains a set~$X_v = \{v\}$.
	Note that the cost of~$\F$ is at most~$k$ and all blue edges in~$H$ are covered by~$X$.
	Moreover, each red edge~$uw$ in~$H$ is resolved as by construction it holds that~$\F$ contains the set~$X_u = \{u\}$ or~$X_w = \{w\}$.
	Without loss of generality, let~$\F$ contain~$X_u$.
	Then, the red edge~$uw$ is resolved as~$w$ is contained in~$X$ and~$u$ is contained in~$X_u \neq X$.
	
	For the other direction, assume that there is an overlapping clustering~$\F$ of~$H$ of cost at most~$k$.
	By \cref{lemma:big-cluster2}, for each vertex~${v \in V}$, there exists a set~$X_v \in \F$ with~$U \cup \{v\} \subseteq X_v$.
	Note that~$U \subset X_u \cap X_v$ for any pair~$u,v \in V$ and therefore~$X_u = X_v$ as otherwise the cost of~$\F$ is at least~$|U| > k$.
	Hence, there exists a set~$X \in \F$ with~$U \cup V \subseteq X$.
	Since all blue edges are covered by~$X$, we next focus on resolving all red edges.
	Note that since~$X$ contains all vertices in~$H$ and the cost of~$\F$ is at most~$k$, all remaining sets in~$\F' = \F \setminus \{X\}$ contain at most~$k$ vertices combined.
	If for some red edge~$uv$ none of the two vertices~$u$ or~$v$ is contained in a set in~$\F'$, then this red edge is not resolved by~$\F$.
	Thus for each red edge, at least one of the two endpoints is contained in a set in~$\F'$.
	Note that this immediately implies that~$G$ contains a vertex cover of size at most~$k$ (all vertices that are contained in a set in~$\F'$).
	This concludes the proof.
\end{proof}

\subsection{Polynomial Kernel}\label{sec:poly-kernel}

We next show that \CEPVSfull{} parameterized by~$k$ admits a polynomial kernel.
Note that this is in stark contrast to the para-\cclass{NP}-hardness of \CCVSfull{} parameterized by~$k$.\looseness-1

\begin{theorem}
    \CEPVSfull{} parameterized by the number of vertex splits admits a kernel with~$O(k^3)$~vertices.
\end{theorem}
\begin{proof}
    Let $(G = (V, B, R),k)$ be the input instance of \CEPVS{}.
    We begin by computing an inclusion-maximal bad star forest $T$ in $G$. If $\weight(T) \geq k$, then we conclude, according to Lemma~\ref{lemma:bad-star-lower}, that $(G,k)$ is a no-instance and output an appropriate trivial kernel.

    Otherwise let $S$ be the vertices of $T$ and note that $|S| \leq 3\weight(T) \leq 3k$. Since $T$ is inclusion-maximal, we know that $G \setminus S$ cannot contain any bad stars and in particular no bad triangles. We conclude that $G \setminus S$ is therefore a cluster graph. Let
    $C_1, C_2, \ldots, C_p$ be these clusters.
    We next exhaustively apply the following simple reduction rule.
    \begin{redrule}\label{red:remove-clusters}
        If $G$ contains a blue clique~$C$ such that
        all edges with one endpoint in~$C$ are red, then remove~$C$
        from $G$.
    \end{redrule}

    Next, we bound the number~$p$ of cliques in~$G\setminus S$ as follows.
    Assume that~$G\setminus S$ contains at least~$4k +1$ cliques.
    Note that by application of \cref{red:remove-clusters}, all clusters in $G\setminus S$ have at least one blue edge towards $S$.
    Pick for each clique~$C$ in~$G\setminus S$ one such blue edge towards~$S$ and let~$v_C$ be the endpoint in~$C$ of this edge.
    Note that these chosen edges form a collection of vertex-disjoint stars with all centers in~$S$ (but not necessarily all vertices in~$S$ being centers).
    Moreover, since the vertices~$v_C$ and~$v_{C'}$ belong to different cliques for each pair~$C \neq C'$ of cliques in~$G\setminus S$, the edge between the two is red.
    Hence, all stars with at least two leaves in~$V \setminus S$ are bad stars.
    Let~$S' \subseteq S$ be the set of vertices in~$S$ that are not the center of such bad stars, that is, vertices in~$S$ for which we chose at most one incident blue edge as a representative for a clique.
    Let~$S^* = S \setminus S'$.
    Since we chose at most one blue edge~$sv_C$ for each vertex~$s \in S'$, the number of chosen blue edges included in bad stars is at least
    \begin{equation*}
    	(4k+1) - |S'| \geq (4k+1) - |S| + |S^*| \geq (4k+1) - 3k + |S^*| = k+1 + |S^*|.
    \end{equation*}
    Hence, the weight of the constructed collection of bad stars is at least~$k+1$ and by \cref{lemma:bad-star-lower}, we conclude that~$(G,k)$ is a no-instance.
    Thus, if~$G \setminus S$ contains at least~${4k+1}$~cliques after applying \cref{red:remove-clusters} exhaustively, we can return a trivial no-instance.
    Otherwise, the number~$p$ of cliques is bounded by~$4k$.

    We are now left with the task of bounding the size of each individual cluster~$C_i$ to arrive at a polynomial kernel.
    To this end, we apply the following marking and deletion procedure to each cluster:
    For a fixed cluster~$C_i$, begin with an initially empty set $M_i$. For each vertex $v \in S$, arbitrarily mark~$k+1$~red and~$k+1$~blue neighbors of~$v$ in~$C_i$ by adding them to~$M_i$ (or all red/blue neighbors if there are at most~$k$).
    Note that we mark at most~$|M_i| \leq |S|(2k+2) \leq 6k^2+6k$ vertices this way.

    \begin{redrule}\label{red:shrink-clusters}
        For each cluster $C_i$, delete all (unmarked) vertices in~${C_i \setminus M_i}$ from~$G$.
    \end{redrule}

    Let $\hat G$ be the graph obtained after applying the reduction rule to some cluster~$C_i$.
    We now need to show that this reduction rule is safe and sound.
    Let~$R_i := C_i \setminus M_i$ be the vertices removed by the reduction rule.

    First note that if $\F$ is an overlapping clustering of $G$, then $\F \setminus R_i$ (interpreted as a multiset\footnote{Technically an overlapping clustering cannot be a multiset. We refer the reader to the proof of~\cref{lemma:bad-star-lower}, in which we formally show how to adapt this construction into an overlapping clustering with cost bounded by $\cost_G(\F)$.}, that is, the same cluster might appear multiple times in it) is trivially an overlapping clustering of $\hat G$ and
    $
        \cost_{\hat G}(\F \setminus R_i) \leq \cost_{G}(\F)
    $.
    Thus, the reduction rule is safe.

    To prove that the reduction rule is sound, let $\hat \F$ be an overlapping clustering of $\hat G$ with~${\cost_{\hat G}(\hat \F) \leq k}$.
    Let~$u \in R_i$ be one of the removed vertices.
    We argue that we can include $u$ in the clustering without increasing the cost.
    Note that since we removed a vertex, the size of~$C_i$ was initially at least~$2k+2$ and hence, by \cref{lemma:big-cluster2}, we have that~$\hat \F$ contains a set $\hat C$ with $C_i \subseteq \hat C$.
    We add~$u$ to this cluster and now argue that $u$ does not have to be included in any further clusters if $(G,k)$ is a yes-instance.
    To this end, we show that every edge incident to $u$ is already covered/resolved by this new clustering.

    Let~$uv$ be any blue edge incident to~$u$.
    Note that if~$v \in C_i$ then $uv$ is covered by $\hat C$, so we may assume that~$v\in S$.
    Then~$v$ has at least~$k+2$~blue neighbors in~$C_i$ as otherwise we would have marked~$u$.
    Let~$N$ be a set of~$k+1$~neighbors of~$v$ in~$C_i$ that were marked.
    By \cref{lemma:big-cluster2}, there exists a cluster set~$X \in \hat\F$ with~$N \cup \{v\} \subseteq X$.
    Hence,~$X = \hat C$ as otherwise the cost of~$\hat\F$ would be at least~$k+1$ as each vertex in~$N$ would appear in at least two sets.
    Thus,~$uv$ is covered by~$\hat{C} \cup \{u\}$ in the constructed overlapping clustering.

    Now let~$uv$ be any red edge incident to~$u$.
    Again, we claim that because~$u$ was unmarked, $v$ must have at least $k+2$ red neighbors in~$C_i$.
    Either~$v \in S$, in which case the argument is the same as before, or $v \in  V \setminus(S \cup C_i)$.
    In this case, all of~$C_i$ is contained in $v$'s red neighborhood, and we have already observed that~$C_i$ has at least $2k+2$ vertices.
    Let~$N$ be a set of~$k+1$ red neighbors of~$v$ in~$C_i$ that were marked.
    Note that~$v$ is contained in a set~$X \neq \hat{C} \in \hat\F$ as otherwise each vertex in~$N$ would be contained in at least two sets and~$\cost(\hat\F)\geq k+1$.
    Hence,~$uv$ is resolved as~$u\in \hat\C \cup \{u\}$ and~$v \in X$.

    We conclude that the resulting clustering covers all blue edges incident to~$u$ and resolves all red edges incident to~$u$ at the same cost as the clustering~$\hat \F$. By  repeating the procedure for the remaining vertices of~$R_i$ we conclude that there exists a clustering~$\F$ which clusters~$G$ and $\cost_{G}(\F) = \cost_{\hat G}(\hat \F)$. Repeating this argument for every cluster demonstrates that Rule~\ref{red:shrink-clusters} is indeed sound.

    Finally, note that after application of Rule~\ref{red:remove-clusters} and Rule~\ref{red:shrink-clusters} to a yes-instance, we have $p \leq 4k$ clusters of size at most $|S|2(k+1) \leq 6k^2+6k$ each and therefore the total number of vertices in the end is at most
    $$
      |S| + 4k(6k^2+6k) = 24k^3 + 24k^2 + 3k \in O(k^3)
    $$ or we return a trivial no-instance. This concludes the proof.
\end{proof}

\subsection{Constant-Factor Approximation}\label{sec:approximation}

We conclude this section with a constant-factor approximation for \CEPVSfull.
Again, this is in stark contrast to \CCVSfull.

\begin{theorem}
    \CEPVSfull{} admits a 7-approximation in polynomial time.
\end{theorem}

\begin{proof}
	Let $G=(V,E)$ be a correlation graph.
    We again begin by computing an inclusion-maximal bad star forest~$T$ in~$G$.
    By \cref{lemma:bad-star-lower}, the weight of~$T$ is at most~$\opt$ (the minimum cost of an overlapping clustering of~$G$).
    Since the number of vertices in a bad star is at most thrice its weight (a bad triangle has weight one and contains three vertices), the set~$S$ of vertices in~$T$ is at most~$3\opt$.
    Since~$T$ is inclusion-maximal, the graph induced by~$V \setminus S$ does not contain any bad star, that is, the blue edges form a cluster graph.
    Let~$\C$ be the set of (blue) cliques in this graph.
    If~$|\C| \leq 1$, then we find a simple 3-approximation by putting each vertex~$v\in S$ into its own cluster set~$X_v$ and adding one cluster set~$X_S = V$.
    Note that the cost of this overlapping clustering is~$|S| \leq 3\opt$.
    Hence, we assume for the remainder of the proof that~$|\C| \geq 2$.

    Next, we restrict our search to a solution that contains one cluster set~$X_S$ with~$S \subseteq X_S$ and for each vertex~$v \in S$ one cluster set~$X_v = \{v\}$.
    Note that we can add these sets to any solution (if they are not already present within the solution) to get a new solution whose cost is at most~$2|S| \leq 6\opt$ larger than the original.
    We call overlapping clusterings that satisfy the above \emph{simple solutions} and we denote the minimum cost of a simple solution by~$\opt'$.

    We next show that there is always a simple solution of cost~$\opt'$ that contains for each clique~$C \in \C$ a cluster set~$X_C$ with~$C \subseteq X_C$.
    Start with any simple solution~$\F$ of cost~$\opt'$ and any clique~$C \in \C$ and assume that~$\F$ does not contain a cluster set containing~$C$.
    We prove that in this case each vertex in~$C$ is contained in at least two cluster sets in~$\F$.
    Assume towards a contradiction that some vertex~$v\in C$ is contained in exactly one cluster set~$Y$ (note that by definition of overlapping clusterings, each vertex is contained in at least one cluster set).
    Since we assumed that no cluster set completely contains~$C$, there exists a vertex~$u \in C \setminus Y$.
    However, since~$u$ and~$v$ are contained in the same clique~$C$, the edge between them is blue and has to be covered by some cluster set~$Z \in \F$ (and hence~$Z$ has to contain both~$u$ and~$v$).
    Note that~$Z \neq Y$ since~$u \in Z$ but~$u \notin Y$.
    This contradicts the assumption that~$v$ is only contained in cluster set~$Y$.

    We construct a new simple solution~$\F'$ of cost~$\opt'$ by removing all vertices in~$C$ from all cluster sets in~$\F$ and adding them all to~$X_S$.
    In addition, we add one new cluster set~$X_C = C$.
    Note that the cost of~$\F'$ is at most the cost of~$\F$ as we removed each vertex in~$C$ from at least two cluster sets and added them to exactly two cluster sets.
    Moreover, the new solution is indeed an overlapping clustering as all blue edges incident to a vertex in~$C$ are covered by~$X_S$ as all blue neighbors are either in~$S$ or in~$C$.
    Since no red neighbors of any vertex in~$C$ are contained in~$C$, the new cluster set~$X_C$ ensures that all red edges incident to vertices in~$C$ are resolved.
    Repeating the above for all cliques in~$\C$ yields a simple solution of cost~$\opt'$ that contains for each clique~$C \in \C$ a cluster set~$X_C$ with~$C \subseteq X_C$.

    The next step is to show that there is always an optimal simple solution (a simple solution of cost~$\opt'$) in which~$X_C \neq X_{C'}$ for any pair~$C \neq C' \in \C$.
    Start with any optimal simple solution~$\F$ that contains a cluster set~$X_C \supseteq C$ for each clique~$C \in \C$ and assume that~$X_{C_1} = X_{C_2}$ for some cliques~$C_1 \neq C_2$.
    Observe that all vertices from at least one of the two cliques are contained in at least two cluster sets each as if there are vertices~$u\in C_1$ and~$v \in C_2$ that are only contained in~$X_{C_1} = X_{C_2}$, then the red edge between them is not resolved by~$\F$.
    Without loss of generality, let all vertices of~$C_1$ be contained in at least two cluster sets each.
    Then, we construct a new simple solution~$\F'$ of cost~$\opt'$ by removing all vertices in~$C_1$ from all cluster sets in~$\F$ and adding them all to~$X_S$ and adding one new cluster set~$X_{C_1} = C_1$.
    The proof that this is correct is exactly the same as before:
    The cost of~$\F'$ is at most the cost of~$\F$ as we removed each vertex in~$C_1$ from at least two cluster sets and added them to exactly two cluster sets.
    Moreover, the new solution is indeed an overlapping clustering as all blue edges incident to a vertex in~$C_1$ are covered by~$X_S$ and the new cluster set~$X_{C_1}$ ensures that all red edges incident to vertices in~$C$ are resolved.
    Repeating the above for all cliques in~$\C$ yields an optimal simple solution that contains for each clique~$C \in \C$ a cluster set~$X_C \supseteq C$ such that~$X_C \neq X_{C'}$ for all~$C \neq C' \in \C$.

    Next, we guess which clique~$C^* \in \C$ satisfies~$X_{C^*} = X_S$ in an optimal simple solution satisfying all of the above.\footnote{We can assume that one such clique always exists, but even if this was not the case, the following proof still works if the guess~$\{C^*\}=\emptyset$ yields an optimal simple solution.}
    By that, we mean that we try all possibilities of the following and return the best solution found.
    For each clique~$C \in \C \setminus \{C^*\}$, we compute a minimum vertex cover~$K_C$ of the blue edges between~$C$ and~$S$ in~$O(n^3)$~time using K{\H o}nig's theorem (note that the considered graph is bipartite by construction).
    We add each vertex in~$K_C \cap S$ to~$X_C$ and each vertex in~$K_C \cap C$ to~$X_S$.

    We claim that~$\sum_{C \in \C \setminus \{C^*\}} |K_C| \leq \opt' - |S|$.
    By the above arguments, we can start with an overlapping clustering~$\F$ consisting of one cluster set~$X_S = S \cup C^*$, one cluster set~$X_C = C$ for each~$C \in \C \setminus \{C^*\}$, and one cluster set~$X_v = \{v\}$ for each~$v \in S$.
    Note that all red edges are resolved and all blue edges except for those between~$S$ and~$V \setminus (S \cup C^*)$ are covered.
    To cover a blue edge~$uv$ with~$u \in S$ and~$v \in C$ for some~$C \in \C \setminus \{C^*\}$, there are three possibilities: We can add~$u$ to~$X_C$, we can add~$v$ to~$X_S$, or there exists a different cluster set~$Y \in F$ with~$\{u,v\} \subseteq Y$.
    Note that in the third case, we can remove~$v$ from~$Y$ and add it to~$X_S$ and still get an optimal simple solution.
    Moreover, the optimal way to cover all blue edges between~$S$ and~$V \setminus (S \cup C^*)$ using the first two possibilities corresponds exactly to~$\sum_{C \in \C \setminus \{C^*\}} |K_C|$.
	Now, all red edges are resolved, all blue edges are covered, and the cost of the constructed overlapping clustering is
	\[|S| + \sum_{C \in \C \setminus \{C^*\}} |K_C| \leq \opt' \leq 7\opt.\]
	Thus, we successfully computed a factor-7 approximation in polynomial time.
\end{proof}

We leave it as an open problem to improve the approximation factor.
We conjecture that a refined concept of a simple solution might yield an approximation factor of~$4$.

\section{Conclusion}
\label{sec:conclusion}

We have introduced permissive vertex splitting, which generalizes the earlier exclusive and inclusive vertex-splitting notions by allowing
symmetry with respect to ``positive'' and ``negative'' pairwise similarity data. Our type of vertex splitting turns out to be quite satisfying
as it corresponds to a natural definition of overlapping clustering. Unfortunately, the general case of
\CCVSfull{} is rather intractable, as it is \cclass{para-NP}-hard and admits no $n^{1-\varepsilon}$-approximation in polynomial time for any~$\varepsilon > 0$ (unless \cclass{P} = \cclass{NP}).
On the positive side, when restricted to data sets with complete data, we obtain a kernel with~$O(k^3)$~vertices and a polynomial-time 7-approximation. Interesting questions remain, for example whether one can reduce our approximation factor to $4$ (or even lower), whether a kernel with~$O(k)$~vertices exists (as is the case when only inclusive splits and edge-edits are allowed~\cite{abu-khzam2023cluster}), or whether refined
lower bounds in terms of running time or approximation factor can be found. Future work might also consider
the parameterized complexity of \CCVSfull{} and \CEPVSfull{} with respect to structural parameters of the input, or with restrictions on the
number of clusters which contain any given node.

Finally, we would like to amplify the call of Abu-Khzam et al.~\cite{abu-khzam2023cluster} to extend the study of vertex
splitting (exclusive, inclusive, or permissive) to other classes of target graphs, many of which (e.g., bicluster graphs, $s$-cliques, $s$-clubs, $s$-plexes, $k$-cores, and $\gamma$-quasi-cliques) have
been proposed as alternatives to cliques in clustering applications.

\bibliographystyle{abbrvnat}
\bibliography{refs}

\end{document}